\documentclass[mathpazo]{cicp}

\usepackage{graphicx}
\usepackage{amssymb}
\usepackage{amsmath}
\usepackage{subfig}
\usepackage[singlelinecheck=off]{caption}
\usepackage{colortbl,color}
\usepackage{algorithm}
\usepackage{amsthm}

\begin{document}
\title{On Stochastic Error and Computational Efficiency of the Markov Chain Monte Carlo Method}


 \author[Li J et.~al.]{Jun Li\affil{1}\comma\corrauth, Philippe Vignal\affil{1,2}, Shuyu Sun\affil{3} and Victor M. Calo\affil{1,3}}
 \address{\affilnum{1}\ Center for Numerical Porous Media, King Abdullah University of Science and Technology, Thuwal, Saudi Arabia. \\
                \affilnum{2}\ Material Science and Engineering, King Abdullah University of Science and Technology, Thuwal, Saudi Arabia. \\
                \affilnum{3}\ Applied Mathematics and Computational Science, Earth Sciences and Engineering, King Abdullah University of Science and Technology, Thuwal, Saudi Arabia}
 \emails{{\tt lijun04@gmail.com} (J.~Li), {\tt philippe.vignal@kaust.edu.sa} (P.~Vignal), {\tt shuyu.sun@kaust.edu.sa} (S.~Sun), {\tt victor.calo@kaust.edu.sa} (V.M.~Calo)}


\begin{abstract}
In Markov Chain Monte Carlo (MCMC) simulations, the thermal equilibria quantities are estimated by ensemble average over a sample set containing a large number of correlated samples. These samples are selected in accordance with the probability distribution function, known from the partition function of equilibrium state. As the stochastic error of the simulation results is significant, it is desirable to understand the variance of the estimation by ensemble average, which depends on the sample size (i.e., the total number of samples in the set) and the sampling interval (i.e., cycle number between two consecutive samples). Although large sample sizes reduce the variance, they increase the computational cost of the simulation. For a given CPU time, the sample size can be reduced greatly by increasing the sampling interval, while having the corresponding increase in variance be negligible if the original sampling interval is very small. In this work, we report a few general rules that relate the variance with the sample size and the sampling interval. These results are observed and confirmed numerically. These variance rules are derived for the MCMC method but are also valid for the correlated samples obtained using other Monte Carlo methods. The main contribution of this work includes the theoretical proof of these numerical observations and the set of assumptions that lead to them.
\end{abstract}


\maketitle

\section{Introduction}
\label{s:intro}
The Monte Carlo method has successfully been applied to a wide variety of applications, which include the solution of integral equations by the Markov Chain Monte Carlo (MCMC) method~\cite{Metropolis1953}, the Boltzmann equation by the Direct Simulation Monte Carlo (DSMC) method~\cite{Bird1994} and stochastic partial differential equations by a multilevel Monte Carlo method~\cite{Barth2012}. We focus our discussion on the MCMC method. An essential part of many scientific problems is to evaluate an integral in a high-dimensional space $\vec X$ with the integrand containing a weighting function $f(\vec X)$ (probability distribution function of the configuration $\vec X$) which is large in some area but close to zero almost everywhere else. The computational cost of evaluating the integral by conventional quadrature schemes is prohibitive since it demands a large number of quadrature points inside a high-dimensional space. This integral can be estimated by the average value of the integrand over a large  number of configurations sampled inside the domain randomly, independently and uniformly, using the Monte Carlo (MC) method. Metropolis and Ulam~\cite{Metropolis1987} (see~\cite{Frenkel2002}) dubbed this simulation method \textit{Monte Carlo} since it uses a large number of random fractions generated by a computer. The accuracy of the MC method can be improved by using the importance sampling scheme~\cite{Marshall1956}, which generates configurations non-uniformly but according to an artificially selected probability density function $g(\vec X)$, which is close to $f(\vec X)$, so that more probability mass is assigned to those configurations with higher probability~\cite{Frenkel2002,Marshall1956,Liu2001}. In order to ensure the sampled configurations remain independent, the process demands the primitive function $G(\vec X)$ of $g(\vec X)$ and its inverse function $\vec X(G)$. Unfortunately, it is not feasible to find such $g(\vec X)$ in most applications of interest. Rather than generating independent configurations, the Metropolis method~\cite{Metropolis1953}, which still uses the importance sampling idea, generates (possibly) correlated configurations from the original $f(\vec X)$ by a Markov chain. The Markov chain makes the algorithm simple and universal. This method is known as MCMC method~\cite{Liu2001}. Since the samples are correlated with each other, the variance of MCMC simulations with the same sample size is larger than the variance of the MC simulations using independent configurations. Additionally, the variance of MCMC simulations usually depends on the sampling interval.

The use of averages is common in scientific studies and many quantities related to thermal equilibria are averaged properties, measured in real experiments over large numbers of particles and long time intervals. If the ergodic hypothesis applies to the system at the molecular level~\cite{Frenkel2002}, we can compute those quantities by ensemble averaging instead of time averaging using the probability distribution function $f(\vec X)$, known from the partition function of the equilibrium state, an idea stemming from statistical mechanics. The MCMC method is a powerful tool based on ensemble averaging idea that can be used to calculate the quantities related to the thermal equilibrium state.

A system with fixed particle number $N$, volume $V$, and temperature $T$ can be described by a canonical ensemble (constant-$NVT$), with the probability distribution function containing only the coordinates of the $N$ particles as independent variables. This description is valid for systems where the quantities of interest only depend explicitly on the location of all the particles. MCMC simulations of this system apply a random sequence of displacements to randomly selected particles. This random selection of particles and displacements is known as a trial move. The sample sequence that it forms generates a (correlated) Markov chain. The correlation degree of this sequence depends on the maximal random displacement applied, that is, the step size that determines the acceptance rate of the trial move.

Most real experiments are carried out under conditions of controlled pressure and temperature. Thus, the isobaric-isothermal ensemble (constant-$NPT$) is widely used in MCMC simulations where the particle location and the volume of the system are randomly modified to visit all possible configurations according to their respective probabilities. Here, the step size of the volume-changing trial move also influences the correlation degree of the successive configurations.

In adsorption studies where the chemical potential $\mu$ is fixed, instead of the particle number $N$, the grand-canonical ensemble (constant-$\mu VT$) is used to calculate the average particle number. The corresponding MCMC method includes a displacement trial move, as well as a trial insertion and removal of particles, with a step size usually fixed to one particle. That is, only one particle is tentatively inserted or removed from the volume each time. The acceptance ratio of particle insertion and removal is very small and thus results in a high-correlation degree of the related successive configurations. This correlation degree cannot be reduced because the step size is already the minimal divisible unit, one particle.

For the simulation of coexisting phases, important in many engineering applications, the MCMC algorithms based on the traditional ensembles described above suffer some important drawbacks. For example, limited computational resources imply that the number of particles used to represent the phase-coexistence system is relatively small. Thus, a large fraction of all particles used reside in the vicinity of the interface between phases. This induces a bias towards the interfacial properties when ensemble averages are computed, rather than including a balanced representation of the bulk phases.

In the literature several improvements to the traditional sampling have been proposed. In~\cite{Panagiotopoulos1987}, a Gibbs-$NVT$ MCMC method, where the total particle number, total volume, and temperature are fixed, was proposed to alleviate these algorithmic restrictions. This Gibbs-$NVT$ scheme combines $NVT$, $NPT$ and $\mu VT$ ensembles for simulating coexisting phases. This combination skillfully avoids the interface predominance by introducing two subsystems modeled as separate boxes. This model allows particles to swap from one phase (box) to the other, while neglecting the potential energy between particles from different phases. Additionally, volume exchanges are allowed between the two boxes while the total volume is conserved. The acceptance ratio of particle swap is very small, as was the case for the grand-canonical ensemble simulation. This limitation can be particularly severe when the density of the dense phase is relatively very high and becomes important when modeling deposition and separation of dense liquids and solids. This drawback is avoided in the Gibbs-Duhem integration method~\cite{Kofke1993a,Kofke1993b,Agrawal1995}. Nevertheless, this integration scheme needs the initial point on the coexistence curve, and thus relies on the use of another method that can provide this initial point. If one of the coexisting phases is a crystal, the method proposed in~\cite{Tilwani1999} improves the acceptance probability of exchanging particles.
\\
\indent The MCMC method based on Gibbs ensemble has successfully been applied to problems related to water systems~\cite{Errington1998}, as well as oil production and processing~\cite{Smit1995,Martin1998,Nath1998,Errington1999,Potoff1999,Ungerer2006,Hajipour2011}. In these applications, the solubility of hydrogen sulfide and other corrosive components in the gas-hydrocarbon mixtures is important data. Nevertheless, this solubility is poorly understood due to the lack of experimental results. In Gibbs-$NVT$ ensemble simulations of two coexisting phases, there are three kinds of trial moves: particle displacement, volume exchange, and particle swap. In order to reduce the variance of the simulation results by decreasing the correlation degree of configurations, we adjust the step size for the first two trial moves. A discussion of the relationship between the variance and the step size of particle displacement is given in~\cite{Frenkel2002} but it is usually difficult to obtain a general rule for such a relationship. Recently~\cite{Li2011}, the liquid-vapor coexistence of methane was simulated by the Gibbs-$NVT$ MCMC method. Then, the variation of mole fraction with pressure in a two-component system at a phase coexistence state was studied with the Gibbs-$NPT$ MCMC method proposed in~\cite{Panagiotopoulos1988}, where the total particle number, pressure, and temperature are fixed.
\\
\indent When Markov chain evolution is used for Monte Carlo simulations, it is not advisable to sample the system for the quantities of interest after each cycle, namely each trial move. Saving a large number of samples to reduce the stochastic noise contained in the samples requires a large amount of memory if the correlation is high; instead, the system is sampled at intervals (sampling interval). The larger the sampling interval is, the smaller the correlation degree of the collected samples will be. The same applies to the variance with fixed sample size (i.e., the total number of sampled cycles). The computational time is almost proportional to the product of the number of samples collected and the sampling interval. Thus, increasing the sampling interval either increases the CPU time when keeping the number of samples constant, or increases the variance of the results when keeping the CPU time constant. Nevertheless, our simulation results show that a good trade-off between the CPU time and memory usage can be achieved. In this paper, we describe the Gibbs-$NVT$ MCMC method and employ it to model the coexisting phases of a Lennard-Jones (L-J) fluid. To make the problem tractable for the following theoretical analysis, we analyze the influence of the sampling interval and sample size on the variance of the simulation results on an idealized fluid, rather than the L-J fluid system. Finally, a general theoretical analysis is proposed to justify and prove some of the empirical observations and rules proposed.
\goodbreak
\section{The Markov Chain Monte Carlo Method}\label{s:MC method}
Let the following integral define the expected value of $A$~\cite{Frenkel2002}:
\begin{equation}\label{eq:<A>=}
    \left<A\right>=\dfrac{\int_{\Omega_{\vec X}}f(\vec X)A(\vec X){\rm d}\vec X}{\int_{\Omega_{\vec X}}f(\vec X){\rm d}\vec X},
\end{equation}
where $\vec X$ is a high-dimensional vector and the formulas of $A(\vec X)$ and $f(\vec X)$ are given. To compute $\left<A\right>$, it is convenient to use the MCMC method to generate correlated configurations $\vec X_i$ after each cycle with a probability density proportional to $f(\vec X)$. Unknown constant coefficients contained in $f(\vec X)$ are canceled in the MCMC computation process. The system is sampled at intervals during the simulation and the configuration $\vec X_j$ at each sampled cycle is used to estimate the expected value $\left<A\right>$ by the average value $\overline{A}=\frac{1}{n}\sum_{j=1}^{n}A(\vec X_j)$ over $n$ samples.
\subsection{Basic algorithm of MCMC method}\label{ss:basic algorithm}
The algorithm of MCMC method~\cite{Metropolis1953} for solving the general integral~\eqref{eq:<A>=} can be summarized as follows:
\begin{enumerate}
\item Initialization of configuration $\vec X$;\label{step1}
\item For each cycle:
  \begin{enumerate}
  \item Apply trial move changing $\vec X$ to $\vec X^\prime$;\label{step2a}
  \item Apply acceptance criterion to the new $\vec X^\prime$;\label{step2b}
  \end{enumerate}
\item Sample the system at regular intervals (after every $d$ cycles);
\item Stop after getting sufficient samples for analysis. \label{step4}
\end{enumerate}
The initial configuration can be selected randomly from within the domain $\Omega_{\vec X}$ of the definition of the configuration space. The Markov chain is generated by randomly modifying the current configuration $\vec X$ into $\vec X^\prime$ using the trial move algorithm.

The algorithm outlined in steps~\ref{step1} to~\ref{step4} should satisfy the ergodicity and time-reversal conditions. The ergodicity condition requires that from the current configuration $\vec X$ it is possible to visit any $\vec X^\prime\in\Omega_{\vec X}$ by a limited number of trial moves. The time-reversal condition requires that the probability for the current configuration to change back to its previous state is larger than zero. The probability density of the trial move event $(\vec X\to \vec X^\prime)$ is denoted by $\alpha(\vec X\to \vec X^\prime)$. Any new configuration $\vec X^\prime$ generated in step~\ref{step2a} will be accepted or rejected in step~\ref{step2b}  based on the following acceptance criterion: $\vec X^\prime$ is accepted if the random number distributed uniformly inside [0,1], $Rf$, is less than \[\dfrac {\alpha(\vec X^\prime\to \vec X)f(\vec X^\prime)}{\alpha(\vec X\to \vec X^\prime)f(\vec X)},\] or rejected otherwise. This means that the acceptance probability is equal to \[acc(\vec X\to \vec X^\prime)=\min\left[1,\dfrac {\alpha(\vec X^\prime\to \vec X)f(\vec X^\prime)}{\alpha(\vec X\to \vec X^\prime)f(\vec X)}\right].\] This selection for the acceptance probability is based on the detailed balance condition for the equilibrium state, which can be stated as \[f(\vec X)\alpha(\vec X\to \vec X^\prime)acc(\vec X\to \vec X^\prime)=f(\vec X^\prime)\alpha(\vec X^\prime\to \vec X)acc(\vec X^\prime\to \vec X),\] and also on the fact that \[\dfrac{\min\left[1, \beta\right]}{\min\left[1, \beta^{-1}\right]}\equiv\beta.\] The algorithm can be simplified significantly by using symmetric trial moves such that the probability density of the trial move from $\vec X$ to $\vec X^\prime$ is equal to the probability density of the reverse move, that is, $\alpha(\vec X\to \vec X^\prime)=\alpha(\vec X^\prime\to \vec X)$. The detailed balance condition is a sufficient but not a necessary requirement, while in~\cite{Manousiouthakis1999} the weaker ``balance condition'' was shown to be a necessary and sufficient requirement.

Samples are collected in step 3 after the simulation has reached the statistical steady state, that is, after an initial transitional period. The quantities of interest are estimated from samples collected every $d$ cycles.
\subsection{MCMC algorithm based on Gibbs-$NVT$ ensemble}\label{ss:Gibbs-NVT algorithm}

\begin{figure}
\centering
  \subfloat[particle displacement: $\vec s_i^\prime=\vec s_i+\Delta s(Rf_1-0.5, Rf_2-0.5, Rf_3-0.5)$][particle displacement: \\ $\vec s_i^\prime=\vec s_i+\Delta s(Rf_1-0.5, Rf_2-0.5, Rf_3-0.5)$]{\includegraphics[width=0.6\textwidth]{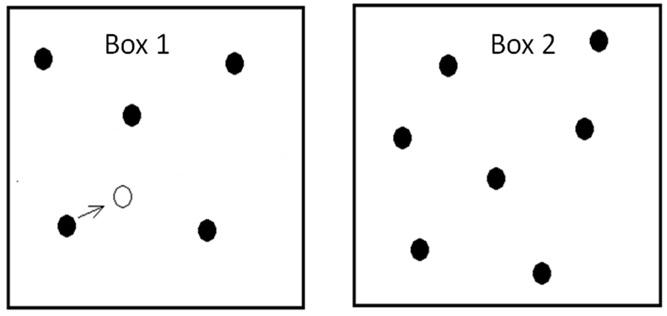}} \\
  \subfloat[volume exchange: $\ln\dfrac{V_1^\prime}{V-V_1^\prime}=\ln\dfrac{V_1}{V-V_1}+\Delta V(Rf_4-0.5), V_2^\prime=V-V_1^\prime$][volume exchange: \\ $\ln\dfrac{V_1^\prime}{V-V_1^\prime}=\ln\dfrac{V_1}{V-V_1}+\Delta V(Rf_4-0.5), V_2^\prime=V-V_1^\prime$]{\includegraphics[width=0.6\textwidth]{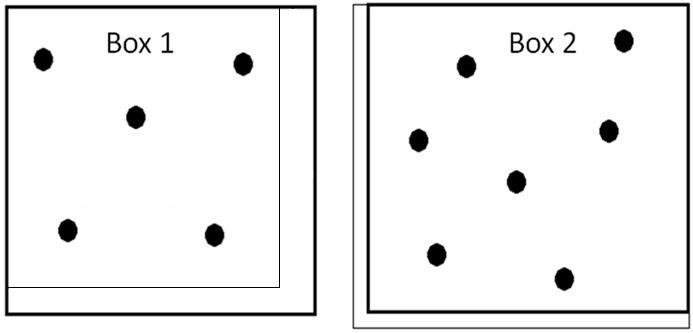}}\\
  \subfloat[particle swap between boxes: $N_1^\prime=N_1\pm1, N_2^\prime=N-N_1^\prime$][particle swap between boxes:\\ $N_1^\prime=N_1\pm1, N_2^\prime=N-N_1^\prime$]{\includegraphics[width=0.6\textwidth]{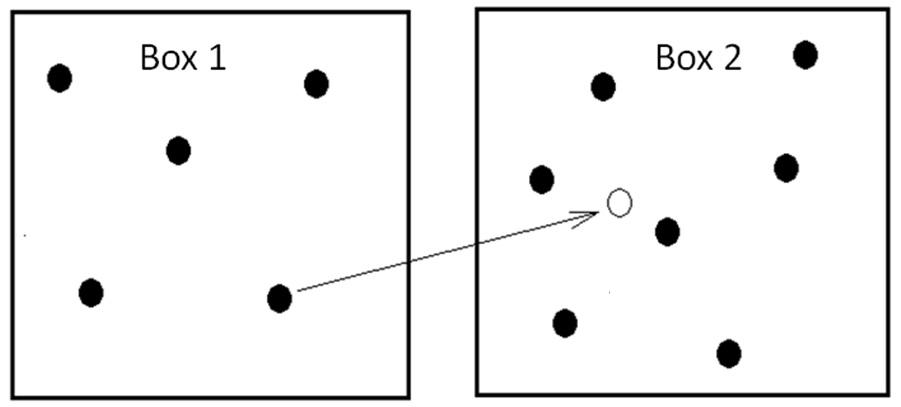}}
  \caption{Schematic model for trial moves with the Gibbs-$NVT$ MCMC method.}
  \label{fig:trialmove}
\end{figure}

We discuss single component systems and assume that each molecule is modeled as a single particle. In the Gibbs-$NVT$ ensemble~\cite{Panagiotopoulos1987}, as described in~\cite{Frenkel2002}, the probability density distribution function $f$ and the related partition function $Q_{\rm G}$ are expressed as
\begin{equation}\label{eq:f=}
    f(N_1,V_1,\vec S_1,\vec S_2)
    =\dfrac{V_1^{N_1}(V-V_1)^{N-N_1}\exp\left[-\beta \left(U_1+U_2\right)\right]}{Q_{G}(N,V,T)V\lambda^{3N}N_1!(N-N_1)!},
\end{equation}
and
\begin{equation}\label{eq:Q=}
    Q_{\rm G}(N,V,T)=\sum_{N_1=0}^N\int_0^V\int\int f(N_1,V_1,\vec S_1,\vec S_2){\rm d}\vec S_1{\rm d}\vec S_2{\rm d}V_1,
\end{equation}
where $T$ is the fixed temperature of both boxes, $V$ is the fixed total volume, $V_1$ is the volume occupied by box 1, $N$ is the fixed total particle number, $N_1$ is the particle number inside box 1, $\vec S_1$ and $\vec S_2$ are high-dimensional vectors that contain the normalized positions $\vec s_i$ of all particles inside boxes $1$ and $2$, respectively, where the normalization parameters are each of the box sizes, which are $V_1^{1/3}$ and $(V-V_1)^{1/3}$, $\lambda=h/\sqrt{2\pi m/\beta}$ is the thermal de Broglie wavelength, $h$ is the Planck constant, $m$ is the molecular mass, $\beta=1/(k_{\rm B}T)$, $k_{\rm B}$ is the Boltzmann constant, and $U_1=U_1(\vec S_1,V_1)$ is the total potential energy of box 1, namely a summation of pair potential energy $u_{ij}$ contributed by particles $i$ and $j$ inside box 1. The probability density distribution function, given in Eq.~\eqref{eq:f=}, and the related partition function, given in Eq.~\eqref{eq:Q=}, are obtained after completing the integration with respect to the momentum variables. Here, the configuration $\vec X$ consists of $N_1$, $V_1$, $\vec S_1$ and $\vec S_2$. In general, we can take $U_1$ as a function $U_1(N_1,V_1,\vec S_1,\vec S_2)$, although it only depends on $N_1$, $V_1$, and $\vec S_1$. Since we have formula~\eqref{eq:f=} for $f(N_1,V_1,\vec S_1,\vec S_2)$, the expected value $\left<U_1\right>$ can be defined as $\left<A\right>$ using Eq.~\eqref{eq:<A>=}. Similarly, we can define the expected values of $\left<U_2\right>$, $\left<p_1\right>$, $\left<p_2\right>$, $\left<V_1\right>$, $\left<V-V_1\right>$, $\left<\rho_1\right>$ and $\left<\rho_2\right>$ using the following definitions of their transient values as functions of $N_1$, $V_1$, $\vec S_1$ and $\vec S_2$. In the MCMC simulations, the successive configurations $\vec X_i$ are generated by a Markov chain according to $f(N_1,V_1,\vec S_1,\vec S_2)$, and the samples of the quantities of interest can be determined from the configurations $\vec X_j$ at the sampled cycles.

The parameter $Q_{G}(N,V,T)V\lambda^{3N}$ in the denominator of Eq.~\eqref{eq:f=} is constant and avoided in the MCMC applications since only the ratio ${f(\vec X^\prime)}/{f(\vec X)}$ is computed to determine $acc(\vec X\to \vec X^\prime)$, as discussed in Section~\ref{ss:basic algorithm}. During the simulation process, $N_1$, $V_1$, $\vec S_1$ and $\vec S_2$ are randomly selected in each cycle and tentatively changed by the corresponding {\textit {symmetric}} trial moves (see Fig.~\ref{fig:trialmove}, where $\Delta s$ and $\Delta V$ are the corresponding step sizes). We compute $acc(\vec X\to \vec X^\prime)=\min\left[1,{f(\vec X^\prime)}/{f(\vec X)}\right]$ using the following formula to avoid the evaluation of $Q_{G}(N,V,T)V\lambda^{3N}$:
\begin{equation}\label{eq:simple f=}
    f(N_1,V_1,\vec S_1,\vec S_2)
    \propto\dfrac{V_1^{N_1}(V-V_1)^{N-N_1}\exp\left[-\beta \left(U_1+U_2\right)\right]}{N_1!(N-N_1)!}
\end{equation}

For Lennard-Jones (L-J) fluids, we have:
\begin{equation}\label{eq:uLJ=}
    u_{ij}=u_{\rm L-J}(r)=4\epsilon\left[\left(\dfrac{\sigma}{r}\right)^{12}-\left(\dfrac{\sigma}{r}\right)^{6}\right]
\end{equation}
where $\epsilon$ is the depth of the potential well, $\sigma$ is the finite distance at which the pair potential energy is zero, and $r=|\vec r_i-\vec r_j|$, where $\vec r_i$ are the coordinates of particle $i$, computed using the normalized $\vec s_i$ as well as the size of the box concerned. To simplify our computations, we replace Eq.~\eqref{eq:uLJ=} by a truncated potential such that
\begin{equation}\label{eq:utrunc=}
    u_{ij}=u^{\rm cut}(r)=
    \begin{cases}u_{\rm L-J}(r), & r\leq r_c; \\
    0, & r>r_c.
    \end{cases}
\end{equation}
An explicit summation of $u_{ij}$ under periodic boundary conditions takes into consideration the infinite periodic images of all particles. Additionally, a correction term due to the contributions beyond the cutoff distance $r_c$ is added to determine the total potential energy for each box. Taking box~1 as an example, the correction for the total energy $U_1$ is~\cite{Frenkel2002}
\begin{equation}\label{eq:Utail=}
    U_1^{\rm tail}=\dfrac{8\pi N_1^2}{3V_1}\epsilon\sigma^3\left[\dfrac{1}{3}\left(\dfrac{\sigma}{r_{c,1}}\right)^9-\left(\dfrac{\sigma}{r_{c,1}}\right)^3\right]
\end{equation}
where $r_{c,1}$ is the cutoff distance for box 1. We use $r_{c,1}=0.45V_1^{1/3}$ and $r_{c,2}=0.45(V-V_1)^{1/3}$, which implies that boxes with different volumes have different cutoff distances. The total potential energy after the tentative trial move at each cycle is computed to determine $f(\vec X')$. The transient pressure, which is computed only at the sampled cycles, can be calculated using the following definition~\cite{Frenkel2002}:
\begin{equation}\label{eq:p=}
    p_1=\dfrac{N_1k_\text{B}T}{V_1}+\dfrac{1}{3V_1}\dfrac{1}{2}{\sum_{i,j,\vec n}}'\left(-\dfrac{{\rm d}u}{{\rm d}r}r\right)
\end{equation}
where the factor 1/2 is used to correct for double counting of the pair-wise contributions and $\vec n$ is a vector of three integers ranging from $(-\infty, \infty)$ through which we can represent the contributions by the infinite periodic particle images. The truncation of Eq.~\eqref{eq:utrunc=} at the cutoff distance is also applied to the explicit summation of Eq.~\eqref{eq:p=} to limit the number of effective pairs. The correction for pressure $p_1$ due to truncation is~\cite{Frenkel2002}
\begin{equation}\label{eq:ptail=}
    p_1^{\rm tail}=\dfrac{16\pi N_1^2}{3V_1^2}\epsilon\sigma^3\left[\dfrac{2}{3}\left(\dfrac{\sigma}{r_{c,1}}\right)^9-\left(\dfrac{\sigma}{r_{c,1}}\right)^3\right]
\end{equation}

In MCMC simulations, it is convenient to use non-dimensional quantities. The resulting non-dimensional system is defined by the following normalized quantities: number density $\rho_1^*=\sigma^3\rho_1=\sigma^3N_1/V_1$, pressure $p_1^*=p_1\sigma^3/\epsilon$, temperature $T^*=Tk_{\rm B}/\epsilon$, and energy $u_{ij}^*=u_{ij}/\epsilon$.

\subsection{MCMC simulations using Gibbs-$NVT$ ensemble}\label{ss:Gibbs-NVT simulation}

\begin{figure}
\centering
\includegraphics[width=0.5\textwidth]{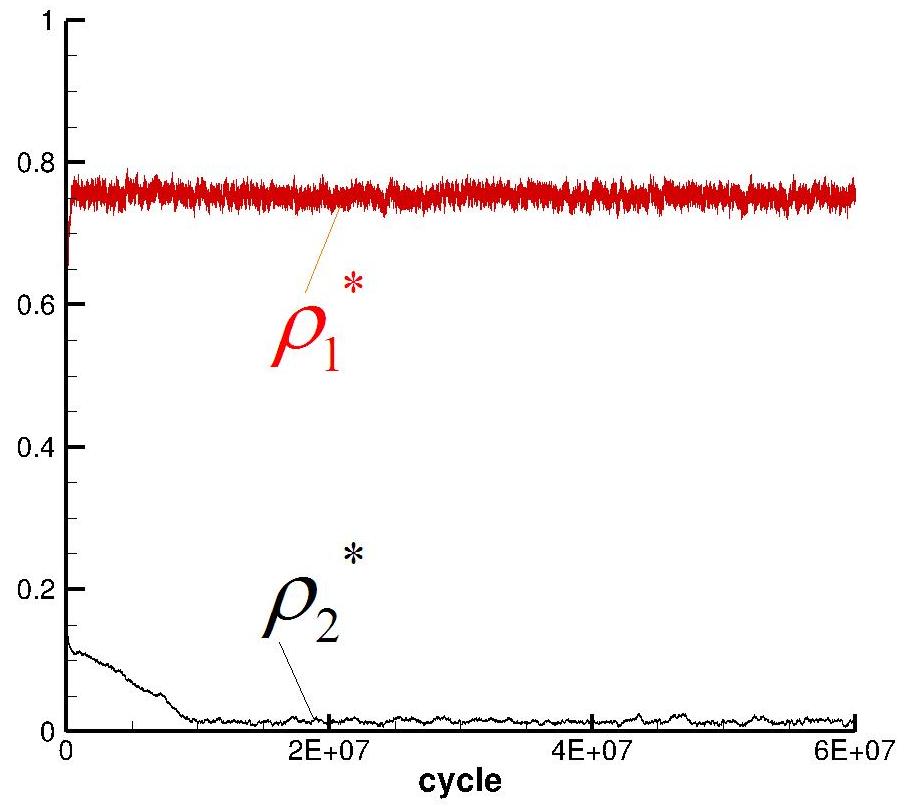}\\
\includegraphics[width=0.5\textwidth]{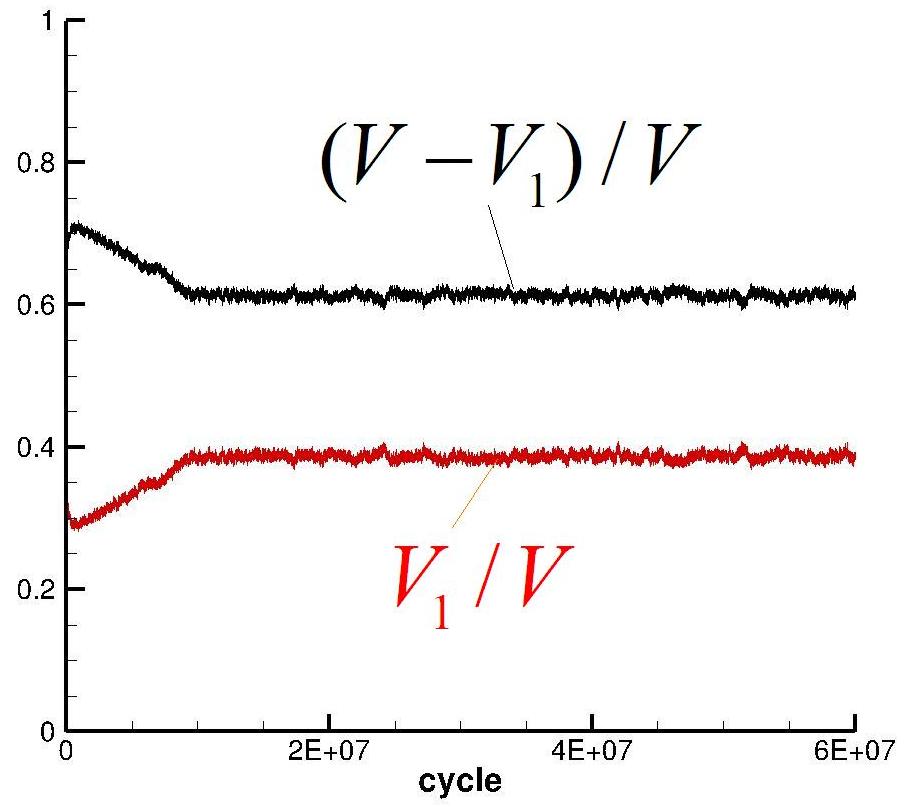}\\
\includegraphics[width=0.5\textwidth]{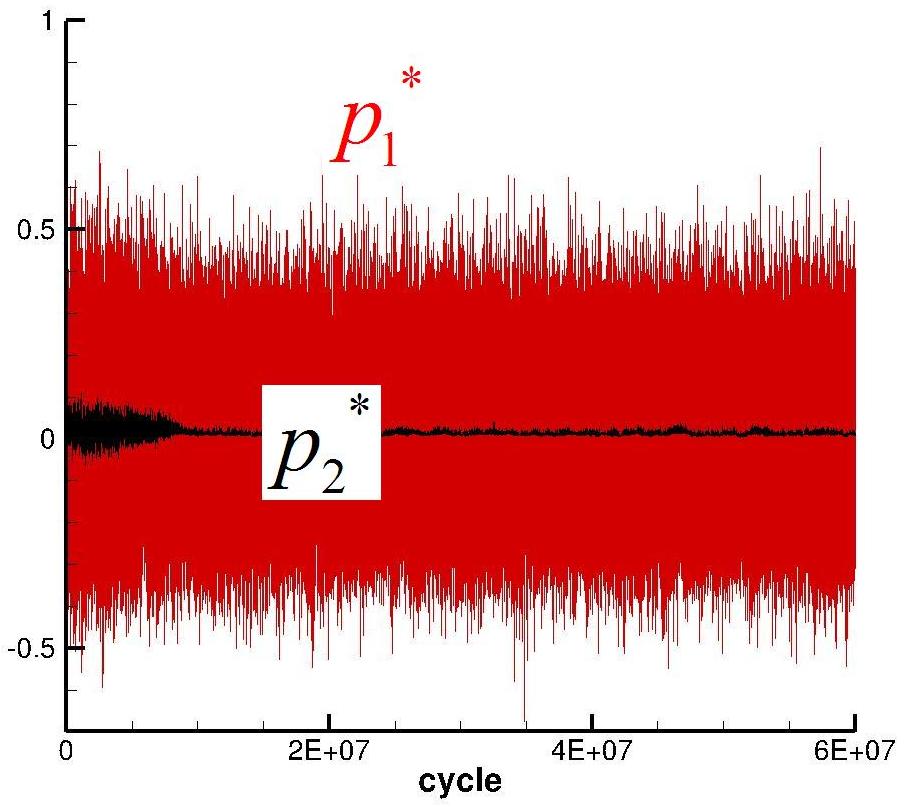}
\caption{Evolution of normalized densities (top), volumes (middle) and pressures (bottom), $T^*=0.9$.}
\label{fig:T=0.9}
\end{figure}

We designed the MCMC code according to the above algorithm based on Gibbs-$NVT$ ensemble and ran the simulations for the phase-coexistence study of a L-J fluid on a Dell workstation (Dell T7500 running Ubuntu 12.04, Intel(R) processor Xeon(R) CPU X5650 @ 2.67GHz, RAM: 47GB). The cutoff distance for the two boxes is fixed at 45\% (smaller than a half) of the corresponding box size, which is modified after each accepted volume exchange trial move. One thousand particles are used in our simulations and the initial normalized density of the two boxes is $\rho_{\rm init}^*=0.3$, unless otherwise stated.

In each cycle, a trial move is applied. It is selected randomly out of three possible cases (displacement move, volume exchange, particle swap, see Fig.~\ref{fig:trialmove} ) that are assigned different probabilities. The probability for selecting the displacement trial move is 0.9, 0.01 for volume exchange and 0.09 for particle swap. After a transitional period (about $L_{\rm init}=2\times10^7$ cycles for the current simulations), we sample the system every $50$ cycles ($d=50$).

The initial values of $\Delta s$ and $\Delta V$ are chosen to be $0.1$ (see Fig.~\ref{fig:adjust,T=0.9}). In order to have the acceptance ratios of the related trial moves be close to user-defined values, the step sizes are modified by an adaptive algorithm (see the source code mentioned in the preface of~\cite{Frenkel2002}) using the collected information. These step sizes are reset at the end of each $L_{\rm adjust}=5\times10^5$ cycles. The adaptive procedure used ensures that, by the completion of the initial $L_{\rm init}$ cycles, the step sizes of the different trial moves are such that the acceptance ratios of those trial moves are approximately equal to the predetermined value (e.g., 0.5 in the current simulations). Once the transitional $L_{\rm init}$ cycles are executed, the step sizes are kept fixed for the remainder of the simulation. Fixing the step sizes ensures the symmetry of the following trial moves.

For $T^*$=0.9, Fig.~\ref{fig:T=0.9} shows that the normalized density, volume, and pressure of the two boxes are converged after the predetermined $L_{\rm init}=2\times10^7$ cycles. Before the $L_{\rm init}$ cycles are complete, the step sizes $\Delta s$ and $\Delta V$ are adjusted, the related achieved acceptance ratios are changed correspondingly, and finally approach the predetermined value of 0.5 as shown in Fig.~\ref{fig:adjust,T=0.9} (left). After $L_{\rm init}$ cycles, the step sizes are fixed to their latest values and the related acceptance ratios fluctuate about 0.5 as desired. Fig.~\ref{fig:adjust,T=0.9} (left) also shows that the acceptance ratio of particle swap between boxes is only about 0.0026 for $T^*=0.9$ because of the very high density of box 1~(see Fig.~\ref{fig:T=0.9}). This situation only worsens as density increases. As discussed in the introduction, this acceptance ratio cannot be improved when $T^*$ is fixed, even though it results in a high-correlation degree of the successive samples. If instead $T^*$ is increased from 0.9 to 1.25, the acceptance ratio of particle swap is increased to about 0.06, as shown in Fig.~\ref{fig:adjust,T=0.9} (right) because the density of the dense phase is decreased. The results of $\rho^*$ for different values of $T^*$ are shown in Fig.~\ref{fig:temp-rho}. They include a comparison with results computed using the equation of state presented in~\cite{Nicolas1979} and MCMC simulations~\cite{Frenkel2002}.

\begin{figure}
\centering
\includegraphics[width=0.49\textwidth]{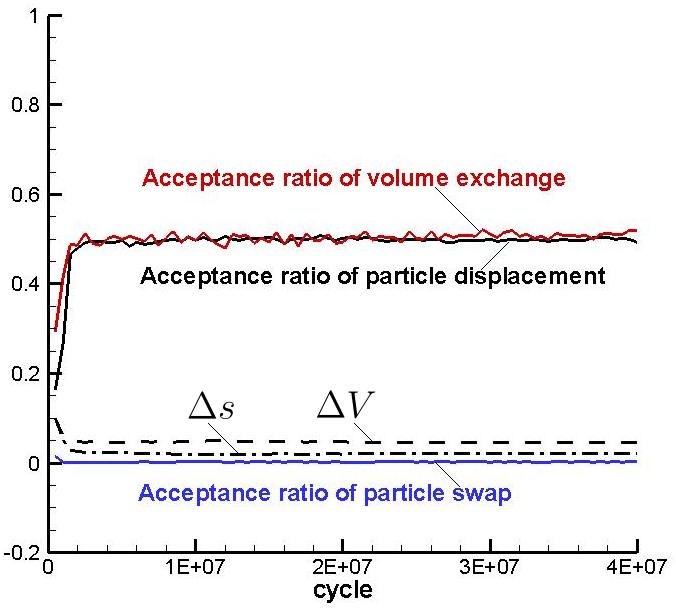}
\includegraphics[width=0.49\textwidth]{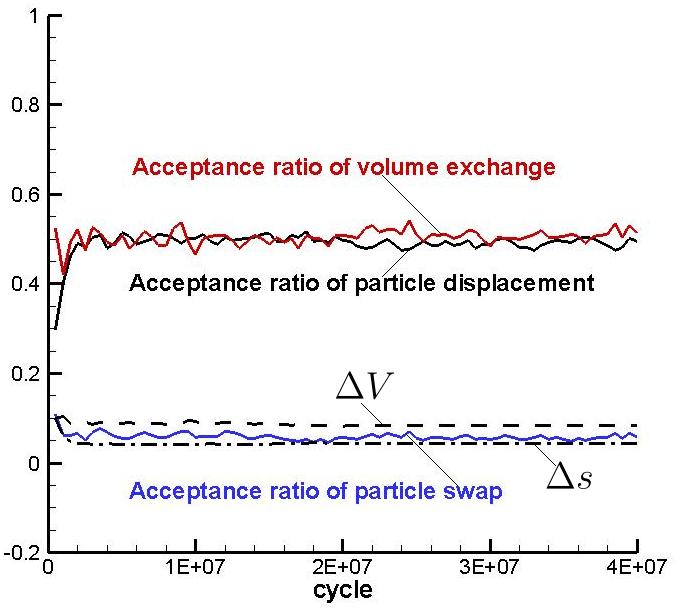}
\caption{Evolution of acceptance ratios and step sizes, $T^*=0.9$ (left) and $T^*=1.25$ (right).}
\label{fig:adjust,T=0.9}
\end{figure}

\begin{figure}
\centering
\includegraphics[width=0.5\textwidth]{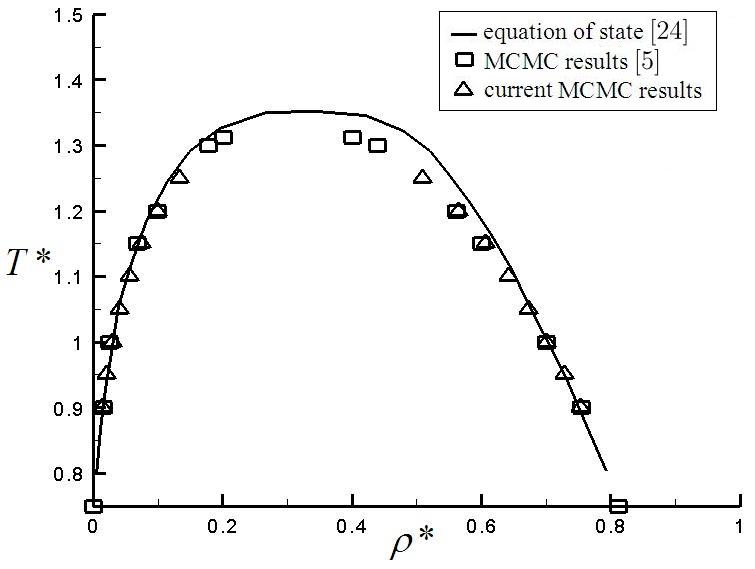}
\caption{Phase diagram of a Lennard-Jones fluid.}
\label{fig:temp-rho}
\end{figure}

\begin{figure}
\centering
\includegraphics[width=0.5\textwidth]{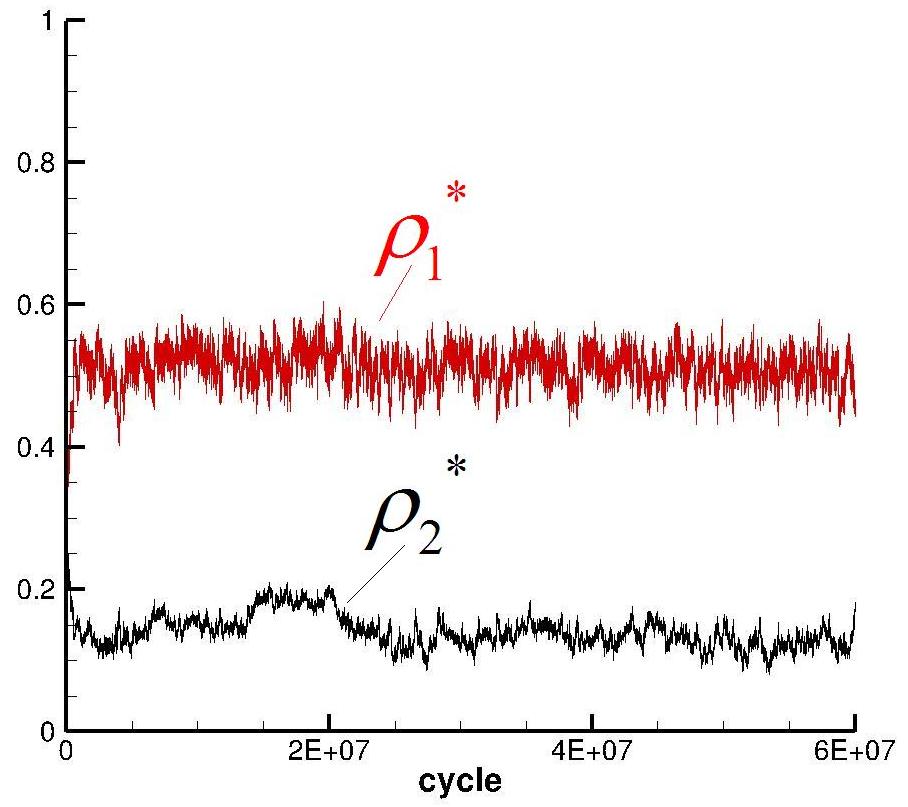}\\
\includegraphics[width=0.5\textwidth]{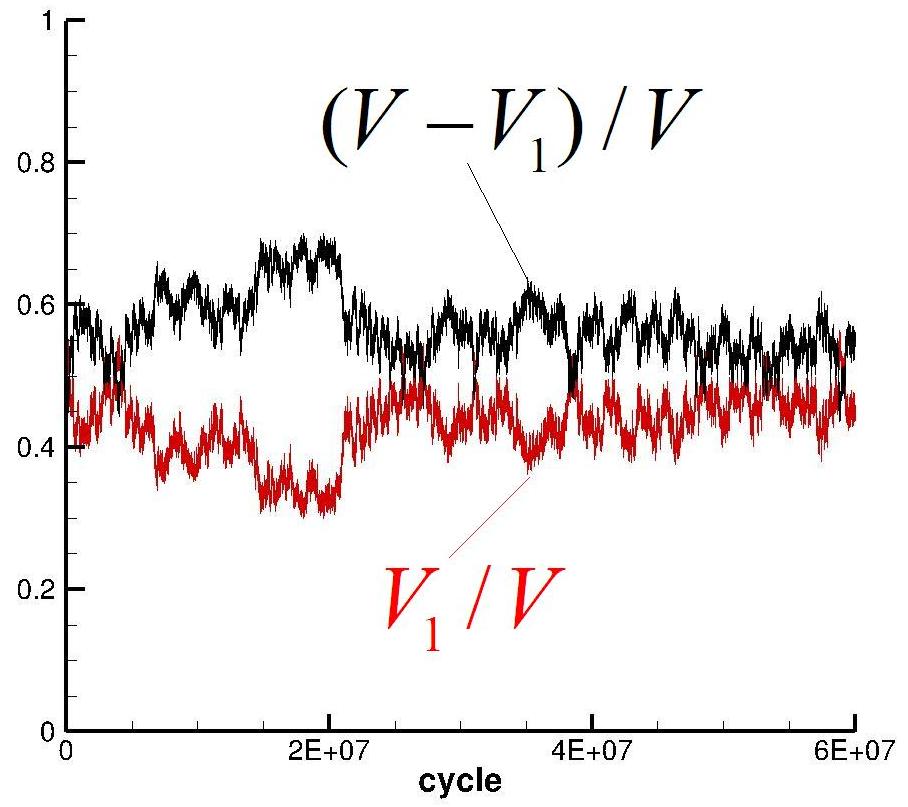}\\
\includegraphics[width=0.5\textwidth]{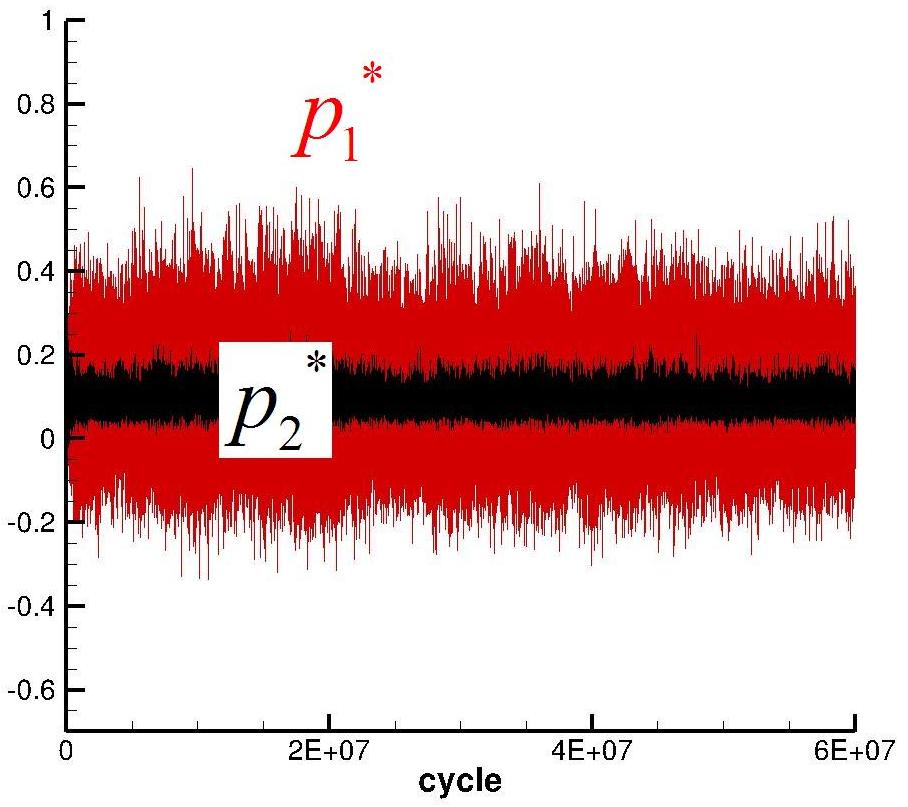}
\caption{Evolution of normalized densities (top), volumes (middle) and pressures (bottom), $T^*=1.25$.}
\label{fig:T=1.25}
\end{figure}

As shown in Fig.~\ref{fig:T=0.9}, the statistical noise of the simulation results is larger in the dense-phase box than in the lower-density box. A similar observation is made in~\cite{Frenkel2002}. For example, the simulation results of $T^*=1.25$ with the same initial density value of 0.3 are shown in Fig.~\ref{fig:T=1.25}, where we observe that the intensity difference of statistical noise of the two phases is reduced by decreasing the density difference.

\section{Blocking Method for Estimating the Variance}\label{s:blocking method}
In the following discussion, we replace $A$ of Eq.~\eqref{eq:<A>=} by $x$, as used in the blocking method described in~\cite{Flyvbjerg1989} to represent the sampled quantities of interest, including pressure, number density, volume, and total potential energy of each box. In MCMC simulations, each sample $x_i$ is a measurement of a random variable $x$ with an exact but unknown probability distribution, from which we define the expected value $\left<x\right>$. We use the average value $\overline x=\dfrac{1}{n}\sum_{i=1}^{n}x_i$ to estimate $\left<x\right>$. This estimation is then unbiased as $\left<\overline x\right>=\left<x\right>$. If the measurements can be taken as independent, the variance $\sigma^2(\overline x)$ of the estimation using $\overline x$ is inversely proportional to the size $n$ of the sample set. But, if they are correlated, the variance then also depends on the sampling interval $d$ between two successive samples.

In the blocking method, the following transformation is employed to decrease the sample size till $n'=2$
\begin{equation}\label{eq:blocking step}
    \begin{cases}x_i'=(x_{2i-1}+x_{2i})/2, \\
    n'=n/2.
    \end{cases}
\end{equation}
After each blocking step, we get a new value for
\begin{equation}\label{eq:c0/(n-1)=}
    \dfrac{c_0'}{n'-1}=\dfrac{1}{(n'-1)n'}\sum_{i=1}^{n'}(x_i'-\overline {x'})^2,
\end{equation}
which increases during the blocking process and approximates $\sigma^2(\overline x)$ if convergence is achieved. $\dfrac{c_0'}{n'-1}$ denotes the value we compute in practice. The value at the convergence point is used to estimate the variance of the average value and this estimation is unbiased~\cite{Flyvbjerg1989}. If the blocking process does not converge, the largest value during the blocking process is a lower bound of the variance~\cite{Flyvbjerg1989}. Convergence happens if the sample set covers a span which is several times larger than the maximal correlation interval $\tau$, so that the ``blocking'' variables $x_i'$ at the convergence point are independent Gaussian variables. The subtlety of the blocking method is to decrease the correlation degree of the new sample set $(x_i')_{i=1,\cdots,n'}$ making the correlated functions $\gamma_{i,j}'\equiv\left<x_i' x_j'\right>-\Big<x_i'\Big>\left<x_j'\right>,i\ne j$ tend to zero.

The definition of $\sigma^2(\overline x)$ is given in Eq.~\eqref{eq:variance=} using the correlation function $\gamma_{i,j}$. An alternative scheme to estimate $\sigma^2(\overline x)$ is to directly select an estimator for $\gamma_{i,j}$. This selection needs to be done carefully since the most obvious estimator for $\gamma_{i,j}$ is a biased one, as its expected value is not exactly equal to $\gamma_{i,j}$~\cite{Flyvbjerg1989}. As shown in~\cite{Flyvbjerg1989}, the estimator of $\sigma^2(\overline x)$ using $\dfrac{c_0'}{n'-1}$ is unbiased since the expected value of $\dfrac{c_0'}{n'-1}$ at the convergence point is equal to $\sigma^2(\overline x)$. Additionally, the blocking method is more efficient than many other estimators of $\sigma^2(\overline x)$~\cite{Flyvbjerg1989}.

\section{Influence of Simulation Parameters on the Variance}\label{s:variance observation}

\begin{figure}
\centering
\includegraphics[width=0.49\textwidth]{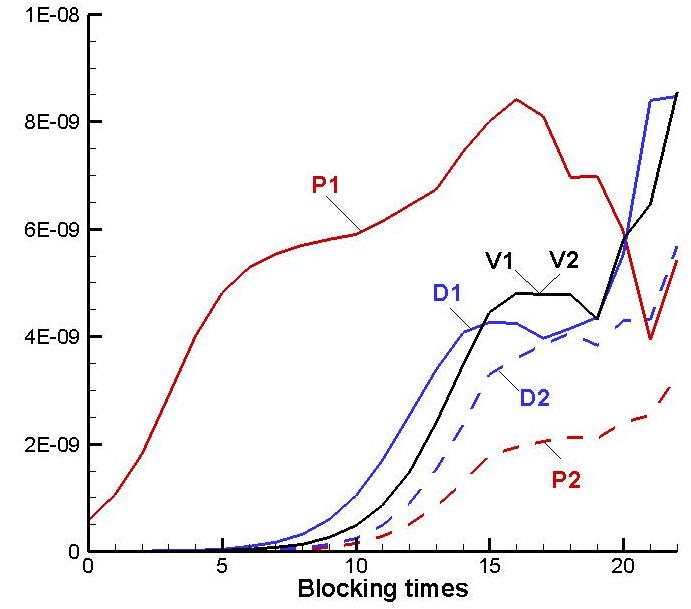}
\includegraphics[width=0.49\textwidth]{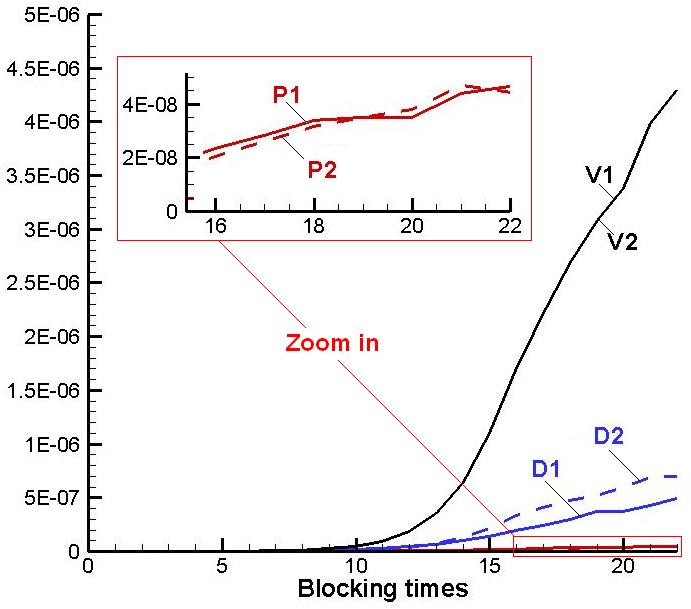}
\caption{Variance estimates by the blocking method,
  $T^*=0.9$ (left) and $T^*=1.25$ (right), $n=2^{25}$ and $d=50$, P$1$ and P$2$ are the variances of average pressures in
  boxes~1 and~2, respectively, D$1$ and D$2$ are the variances of average densities
  in boxes~1 and~2, respectively, and V$1$ and V$2$ are the variances of average volumes in
  boxes~1 and~2, respectively.}
\label{fig:variance,T=0.9}
\end{figure}

We take the set of samples after each cycle as the full sample set. If the trial move at the current cycle is accepted, the current sample is different from the previous one. If the trial move is instead rejected, the configuration remains unchanged and the current sample is the same as the previous one. The repeated samples induce a high correlation degree in the sample set, and  are reasonable from a statistical point of view. Unfortunately, they only contain little useful information. The lower the correlation degree is, the smaller the variance with a given sample size will be. Instead of sampling after each cycle, we could, for example, add a sample to the set after each $d$ cycles. The new sample set will be referred to as coarse sample set, which is a subset of the full sample set. We can reduce the correlation degree of the coarse sample set by increasing $d$. The total number of samples in this coarse set is denoted by $n$. In MCMC simulations, only the coarse sample set is stored, and the memory or disk usage can be reduced significantly by having $d$ be much larger than one. The average value and the corresponding variance are calculated using the coarse sample set.

In the above simulation of a L-J fluid with $T^*=0.9$ in Fig.~\ref{fig:T=0.9}, we observed more statistical noise in box $1$, which has the denser phase. Fig.~\ref{fig:variance,T=0.9} (left) also shows that the variances, estimated by the blocking method, of the normalized density and  pressure of box $1$, are larger than those of box 2 (the final wild fluctuation is due to numerical instabilities when $n'$ becomes very small). Their volume variances are the same since the total volume $V$ is fixed, which is consistent with the data shown in Fig.~\ref{fig:T=0.9}. The relative differences in variance of the number density and the pressure between the two boxes are reduced in Fig.~\ref{fig:variance,T=0.9} (right) compared to these in Fig.~\ref{fig:variance,T=0.9} (left) due to the increase of $T^*$, which is consistent with the comparison between Fig.~\ref{fig:T=1.25} and Fig.~\ref{fig:T=0.9}.

Now, we discuss the variances of the simulation results of $T^*=1.25$. Fig.~\ref{fig:T=1.25} implies that $\sigma^2(\rho_1)$ is larger than $\sigma^2(\rho_2)$ but Fig.~\ref{fig:variance,T=0.9} (right) shows that the variance $\sigma^2(\overline\rho_1)$ of $\overline\rho_1$ of the dense phase is smaller than $\sigma^2(\overline\rho_2)$, which is different from the observation of $T^*=0.9$ where the dense phase has a larger variance. Fig.~\ref{fig:variance,T=0.9} (right) also shows that $\sigma^2(\overline p_1)$ is smaller than $\sigma^2(\overline\rho_1)$, although $\sigma^2(p_1)$ is larger than $\sigma^2(\rho_1)$ as shown in Fig.~\ref{fig:T=1.25} where the variation of $p_1$ ranges from about -0.2 to 0.4, and $\rho_1$ varies from about 0.45 to 0.55. Eq.~\eqref{eq:variance=} gives the definition of the variance $\sigma^2(\overline x)$ of $\overline x$ as a summation of the correlation functions $\gamma_{i,j}$, which can be replaced by $\gamma_t$ where $t=|i-j|$ is the interval between the two samples of $x_i$ and $x_j$. Although $\gamma_0(p_1)>\gamma_0(\rho_1)$, as $\gamma_0(p_1)=\sigma^2(p_1)$ and $\gamma_0(\rho_1)=\sigma^2(\rho_1)$, the decay speed of $\gamma_t(p_1)$ with the increase of $t$ is much faster than that of $\gamma_t(\rho_1)$, as shown in Fig.~\ref{fig:T=1.25}, where $\rho_1$ has periodic fluctuations of scales larger than those observed in the fluctuations of $p_1$. Thus, $\sigma^2(\overline p_1)$ can be smaller than $\sigma^2(\overline\rho_1)$ even though $\sigma^2(p_1)>\sigma^2(\rho_1)$ according to Eq.~\eqref{eq:variance=}. A similar interpretation applies to the observation of $\sigma^2(\overline\rho_1)<\sigma^2(\overline\rho_2)$.

The CPU time is proportional to the total cycle times $L_{\rm total}$, which is almost equal to $n\times d$ ($L_{\rm total}=L_{\rm init}+n\times d$, but the cycle times $L_{\rm init}$ before convergence is negligible). We discuss the influence of $n$ and $d$ on the variance in what follows. The rules that we obtain are expected to be independent of the particular MCMC simulation used to generate the correlated sample set. Therefore, an ideal system, which is simpler than the L-J system and makes the simulation more efficient, is used in the following simulations.

\begin{figure}
  \centering
  \includegraphics[width=0.49\textwidth]{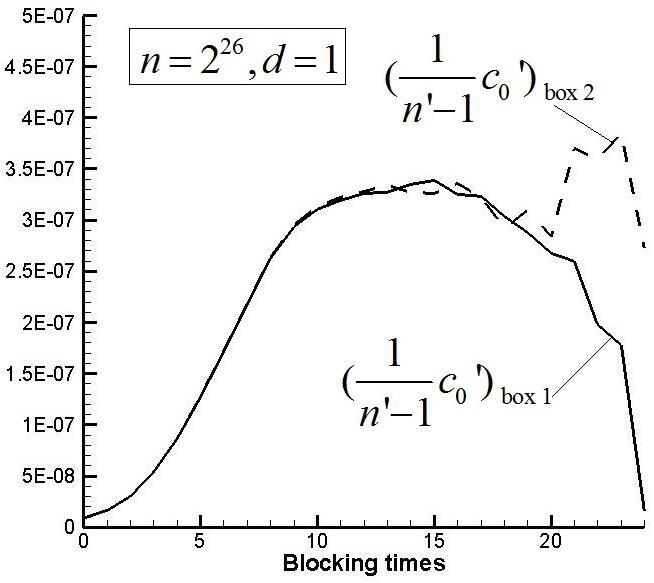}
  \includegraphics[width=0.49\textwidth]{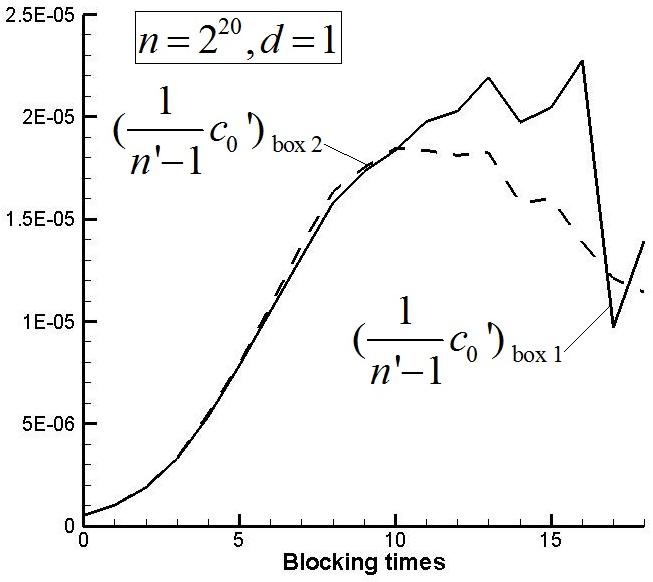}
  \caption{Blocking process of energy sample sets from the two boxes of the ideal system.}
  \label{fig:Ising variance}
\end{figure}

In the ideal system, the total particle number $N$ is equal to $12$, and particle coordinates only take integral numbers $s_i=\pm 1$ as in the Ising model. The probability distribution function becomes
\begin{equation}\label{eq:f(Ising)=}
    f(N_1, \vec S_1, \vec S_2)\propto\dfrac{\exp\left[-(U_1+U_2)\right]}{N_1!(N-N_1)!}
\end{equation}
where the total potential energy $U_1=-J\sum_{i=1,\cdots,N_1 \atop i<j\leq N_1}s_is_j$ is a summation over all pairs (the periodic images of the particles are neglected here) located inside box 1 (the same for $U_2$), and $J=0.1$, so that the acceptance ratio is not too small. For this model, we only need the trial move of particle swap and the spin trial move which randomly selects a particle and changes the sign of $s_i$. The properties of the two boxes are equivalent, thus the correlation degree of their sample sets is the same. We use the energy samples of $U_1$ and $U_2$ as $x_i$ to compute $\dfrac{c_0'}{n'-1}$ using Eq.~\eqref{eq:c0/(n-1)=}. The MCMC simulation results show that the blocking processes of the sample sets of $U_1$ and $U_2$ are very similar as shown in Fig.~\ref{fig:Ising variance}. During the blocking process, only the evolution of $\dfrac{c_0'}{n'-1}$ at the initial stage provides useful information relevant to the approaching process to the variance. As $n'$ shrinks, the evolution of the blocking process becomes unstable, leading to wild fluctuations that can be arbitrarily large, either increasing or decreasing the computed value. These oscillations are due to numerical instabilities and only serve to bound the trustworthy region of the blocking computation. These instabilities do not cause any problem if $\dfrac{c_0'}{n'-1}$ converges before losing stability, as seen in Fig.~\ref{fig:Ising variance} (left). Thus, the value at the convergence point can be used to estimate the corresponding variance. But, in some cases where $\dfrac{c_0'}{n'-1}$ cannot converge before losing stability (see Fig.~\ref{fig:Ising variance}~(right)), it is difficult to judge where the separation point of the two stages is located. That is, the lower bound of the variance, namely the largest value before losing stability, is unknown. When using two sample sets with the same correlation degree, their initial stages should be the same while their final stages are random, which makes the separation point of the two stages easy to find. As shown in Fig.~\ref{fig:Ising variance}~(left) using $n=2^{26}$ samples, the two curves overlap with each other and deviate after blocking 14 times which is the separation point. As it converges before the separation point, the variance for these two sample sets is about $3.3\times10^{-7}$. In Fig.~\ref{fig:Ising variance}~(right), while using $2^{20}$ samples, the two curves overlap with each other before blocking 10 times (the separation point) but are still not converged. This gives us the lower bound of the variance, which is about $1.84\times10^{-5}$. This is the largest value achieved before losing stability. Using different sample sets with similar correlation degrees simplifies the computation of the lower bound of the variance. Nevertheless, in real MCMC simulations, this would incur in prohibitive computational demands in terms of memory usage and CPU time. Thus, as shown in Fig.~\ref{fig:Ising variance}, we propose to use the first maximal point in the blocking process as the separation point and to estimate the lower bound of the variance. This observation is justified by the fact that $\dfrac{c_0'}{n'-1}$ is a theoretically non-decreasing quantity, while the oscillations shown in Figs.~\ref{fig:variance,T=0.9} and~\ref{fig:Ising variance} can be justified by the loss of stability in the blocking computation.

\begin{table}
\centering
\caption{Variance of Markov Chain Monte Carlo simulation results with different sample size $n$ and sampling interval $d$.}\label{tab:variance(n,d)}
\begin{tabular}{ccccc}
\hline
$$       & $n=2^{26}$          & $n=2^{24}$          & $n=2^{22}$          & $n=2^{20}$          \\
\hline
$d=2^0 $   & $3.3\times 10^{-7}$ & $1.3\times 10^{-6}$ & Not converged        & Not converged   \\
$d=2^2 $   & $8.2\times 10^{-8}$ & $3.3\times 10^{-7}$ & $1.4\times 10^{-6}$ & Not converged   \\
$d=2^4$    & $2.3\times 10^{-8}$ & $9.1\times 10^{-8}$ & $3.6\times 10^{-7}$ & $1.5\times 10^{-6}$ \\
$d=2^6$    &          /                      & $4.3\times 10^{-8}$ & $1.7\times 10^{-7}$ & $6.8\times 10^{-7}$ \\
\hline
\end{tabular}

\small{*Note: ``Not converged'' refers to simulations where the blocking process becomes unstable before achieving a definite maximum, as shown in Fig.~\ref{fig:Ising variance}~(right).}
\end{table}

Table~\ref{tab:variance(n,d)} displays the variance for different combinations of $n$ and $d$. The rows of table~\ref{tab:variance(n,d)} correspond to a fixed sampling interval, $d$, which implies that the correlation degree for the coarse sample set is also fixed. As observed in table~\ref{tab:variance(n,d)}, for fixed $d$, the variance is almost inversely proportional to the sample size $n$. This feature is well known for independent sample sets but deserves further theoretical analysis for a general sample set. The CPU time is almost proportional to $n\times d$ as mentioned before. Thus, the variance with fixed $d$ is also almost inversely proportional to the CPU time and so, it is the most rewarding choice for reducing the variance to increase $n$ in view of CPU time. Using $V(n,d)$ as the variance at $(n,d)$, we observed in table~\ref{tab:variance(n,d)} that:
\begin{equation}\label{eq:obser1}
    \dfrac{V(n_2,d)}{V(n_1,d)}=\dfrac{n_1}{n_2}
\end{equation}
Although increasing $n$ is an efficient choice for reducing the variance in view of CPU time, it has an onerous cost for memory or disk usage. For $d=1$, the variance becomes very small only if $n$ is very large which makes the memory requirement unacceptable. In order to reduce the variance while keeping the memory or disk usage low, we decrease the correlation degree of the coarse sample set by increasing $d$. For $n=2^{24}$, the variance decreases to about a quarter of the previous value when $d$ increases from one to four, which is almost as efficient as increasing $n$ in view of CPU time, also increased four times. But, when $d$ increases from 16 to 64 with the CPU time being increased four times again, variance is reduced to 0.47 times the previous value, instead of 0.25 times as the ideal value, from $9.1\times 10^{-8}$ to $4.3\times 10^{-8}$. This is wasteful with regard to CPU time, because we can choose to increase $n$ from $2^{24}$ to $2^{26}$ while fixing $d$ at 16, with CPU time increasing by four times too, but with the variance decreasing to about 0.25 times the previous value ($9.1\times 10^{-8}$ to $2.3\times 10^{-8}$), as already pointed out above. The following theoretical analysis can further prove that
\begin{equation}\label{eq:obser2}
    \dfrac{d_1}{d_2}<\dfrac{V(n,d_2)}{V(n,d_1)}\leq 1,
\end{equation}
where $d_2>d_1$. The equality holds when the samples of the coarse sample set of $d_1$ are already independent and so the correlation degree of the coarse sample set cannot be further reduced by increasing $d$.

Usually, we also want to know how to reduce the variance for a given CPU time, namely $n\times d$. In the case of small $d$, the magnitude of the variance is more dependent on the CPU time. The larger the CPU time is, the smaller the variance will be. Given the same CPU time, the larger the sample size $n$ is, the smaller the variance will be. In the case of large $d$ though, the magnitude of the variance depends more on the sample size $n$ and in the limit case, becomes independent of $d$. In fact, these rules are nothing new compared with Eqs.~\eqref{eq:obser1}-\eqref{eq:obser2}, from which we have that
\begin{equation}\label{eq:obser3}
   \dfrac{V(n_2,d_2)}{V(n_1,d_1)}=\dfrac{n_1V(n_2,d_2)}{n_2V(n_2,d_1)},
\end{equation}
such that
\begin{equation}\label{eq:obser4}
    \dfrac{n_1d_1}{n_2d_2}<\dfrac{V(n_2,d_2)}{V(n_1,d_1)}\leq \dfrac{n_1}{n_2}.
\end{equation}
For $n_1\times d_1=n_2\times d_2$ corresponding to the same CPU time, Eq.~\eqref{eq:obser4} can be replaced by a special form using a new variable $V'({\rm CPU time}, d)$, such that
\begin{equation}\label{eq:obser5}
    1<\dfrac{V'({\rm CPU time}, d_2)}{V'({\rm CPU time}, d_1)}\leq \dfrac{n_1}{n_2},
\end{equation}
where $d_2>d_1$, as required in Eq.~\eqref{eq:obser2}.

The maximal correlation interval $\tau$ of the full sample set can be estimated by the sampling interval $d$ and the blocking times before convergence, because the 'blocking' variables $x_i'$ at the convergence point are independent Gaussian variables~\cite{Flyvbjerg1989}. For the case of Fig.~\ref{fig:Ising variance} (left) with $n=2^{26}$ and $d=1$, the blocking process converges after blocking about 11 times. The estimation of $\tau$ is thus $2^{11}\times d=2^{11}$. In table~\ref{tab:tau(n,d)}, we present the estimates of $\tau$ for different sampling intervals $d$ and sample sizes $n$, reported in table~\ref{tab:variance(n,d)}. The data shows that when the blocking processes converge, different $n$ and $d$ lead to similar estimates of $\tau$, with a value close to $2^{11}$. This is to be expected as we are using different $d$ and $n$ to sample the same random experiment, where the correlation degree of the full sample set is fixed. For the variance analysis of table~\ref{tab:variance(n,d)}, all coarse sample sets satisfy the conditions of $d\ll\tau$ and $nd\gg\tau$, which are required in the following theoretical analysis of the relationship between variance and the sampling parameters $d$ and $n$.

\begin{table}
\caption{Blocking times before convergence used by different coarse sample sets from the same random experiment}\label{tab:tau(n,d)}
\centering
\begin{tabular}{ccccc}
\hline
$$              & $n=2^{26}$ & $n=2^{24}$ & $n=2^{22}$      & $n=2^{20}$          \\
\hline
$d=2^0 $   & $11$            & $11$           & Not converged & Not converged   \\
$d=2^2 $   & $9 $             & $8 $           & $9$                   & Not converged   \\
$d=2^4$    & $8 $             & $7 $           & $7$                   & $7$                 \\
$d=2^6$    &  /                  & $5 $           & $5$                   & $5$                 \\
\hline
\end{tabular}
\end{table}

\section{Theoretical Analysis}\label{s:theoretical analysis}

In section~\ref{s:variance observation}, we describe some empirical rules between the variance and the sample size $n$ and sampling interval $d$, namely Eqs.~\eqref{eq:obser1}-\eqref{eq:obser2}. These rules are independent of the blocking method used to calculate the variance and reflect the underlying feature of the statistical rules, which are independent of the Monte Carlo methods used to generate the correlated samples. The theoretical analysis in this section justifies these rules.

\begin{figure}
\includegraphics[width=1.0\textwidth]{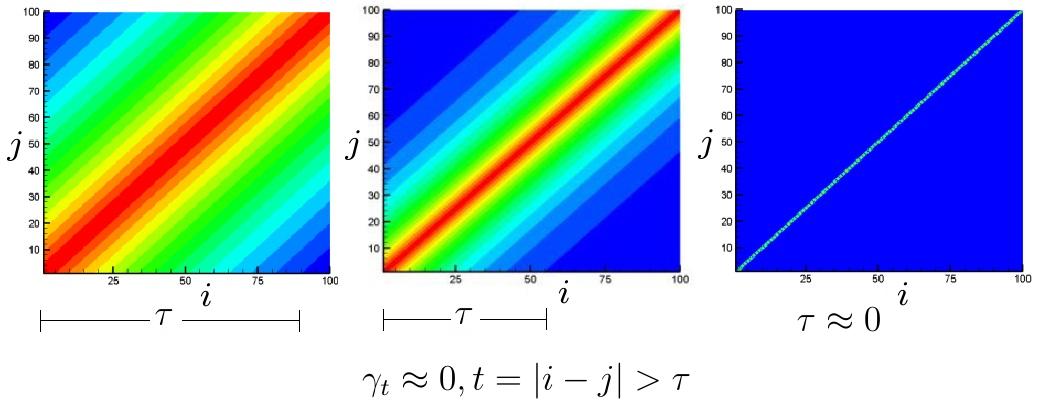}
\caption{$\gamma_t$ of representative sample sets with different correlation degrees.}
\label{fig:correlation degree}
\end{figure}

Let $x_1,x_2,\cdots,x_n$ be the full sample set of a random variable $x$ in a MCMC simulation at thermal equilibrium, which has the following features~\cite{Flyvbjerg1989}:
\begin{align}
  \label{eq:<feature>}
  \begin{cases}
    \left<x_i\right>=\left<x_{i+t}\right>, & \forall t \\
    \left<x_ix_j\right>-\left<x_i\right>\left<x_j\right>
    =\left<x_lx_m\right>-\left<x_l\right>\left<x_m\right>, & |i-j|=|l-m|
  \end{cases}
\end{align}
where $\left<\cdots\right>$ denotes the expected value with respect to these exact but unknown probability distributions. In MCMC simulations, we estimate the expected value $\left<x\right>$ by the average quantity $\overline x=\dfrac{1}{n}\sum_{i=1}^{n}x_i$. The variance of $\overline x$ is~\cite{Flyvbjerg1989}:
\begin{equation}\label{eq:variance=}
    \sigma^2(\overline x)=\left<\overline x^2\right>-\left<\overline x\right>^2=\dfrac{1}{n^2}\sum_{i,j=1}^n\gamma_{i,j}=\dfrac{1}{n}\left[\gamma_0+2\sum_{t=1}^{n-1}\left(1-\dfrac{t}{n}\right)\gamma_t\right],
\end{equation}
where $\gamma_{i,j}=\left<x_ix_j\right>-\left<x_i\right>\left<x_j\right>$ and $\gamma_t\equiv\gamma_{i,j},t=|i-j|$. The variance $\sigma^2(\overline x)$ of the average value $\overline x$ differs from the variance $\sigma^2(x)$ of the random variable $x$ itself. We have $\sigma^2(x)=\gamma_{i,i}=\gamma_0$, which is a fixed value for a given random variable $x$, with a fixed probability distribution, while $\sigma^2(\overline x)$ depends on $n$ and $d$. We define the maximal correlation interval $\tau$ for the full sample set as $\gamma_t\approx0$, where $t>\tau$ (see Fig.~\ref{fig:correlation degree}). In MCMC simulations, it is reasonable to assume that
\begin{equation}\label{eq:assumption}
    \dfrac{1}{n^2}\sum_{i,j=1 \atop i\ne j}^n\gamma_{i,j}=\sigma^2(\overline x)-\dfrac{1}{n}\sigma^2(x)\geq 0,
\end{equation}
where the equality holds when the samples are independent of each other. Fig.~\ref{fig:correlation degree} shows some representative results of $\gamma_{i,j}$ in usual MCMC simulations. Fig.~\ref{fig:correlation degree}~(left) shows the results of a high-correlation sample set compared to Fig.~\ref{fig:correlation degree} (middle). In the limit case where all samples are independent, $\gamma_{i,j}$ is equal to a constant $\sigma^2(x)$ for $i=j$ and zero otherwise, as shown in Fig.~\ref{fig:correlation degree}~(right). We use these schematic models only to show the contour distributions, the monotone interval and the location of maximal value. These models make it easy to understand the following linear interpolation scheme.

\begin{theorem}
In MCMC simulations, the correlation degree of the full sample set is given, thus $\tau$ is fixed. For a general coarse sample set with fixed $d$, the variance $\sigma^2(\overline x)$ is inversely proportional to $n$ if $n\times d\gg\tau$.
\end{theorem}
\begin{proof}
We introduce
$Y_b=\left\{y_{b,i}\left|y_{b,i}=x_{(i-1)d+1},i=1,2,\cdots,n\right.\right\}$
containing $n$ samples generated once from each $d$ cycles. For an arbitrary $d$, Eq.~\eqref{eq:variance=} is modified to define the variance $\sigma^2(\overline y_b)$ as
\begin{equation}\label{eq:proof1-1}
\begin{aligned}
    \sigma^2(\overline y_b)&=\dfrac{1}{n^2}\sum_{i=(k_i-1)d+1 \atop {j=(k_j-1)d+1 \atop k_i,k_j=1,\cdots,n}}\gamma_{i,j}\\
    &=\dfrac{1}{n}\left[\gamma_0+2\sum_{t=1}^{n-1}\left(1-\dfrac{t}{n}\right)\gamma_{td}\right]\\
    &\approx\dfrac{1}{n}\left[\gamma_0+2\sum_{t=1}^{\tau/d}\left(1-\dfrac{t}{n}\right)\gamma_{td}\right]
\end{aligned}
\end{equation}
Assuming $nd\gg\tau$, making $\dfrac{\tau/d}{n}\ll1$, we conclude that
\begin{equation}\label{eq:proof1-2}
    \sigma^2(\overline y_b)\approx\dfrac{1}{n}\left[\gamma_0+2\sum_{t=1}^{\tau/d}\left(1-\dfrac{t}{n}\right)\gamma_{td}\right]
    \approx\dfrac{1}{n}\left[\gamma_0+2\sum_{t=1}^{\tau/d}\gamma_{td}\right]
\end{equation}
which is inversely proportional to $n$ and consistent with Eq.~\eqref{eq:obser1}.
\end{proof}
\newtheorem*{myremark}{Remark 5.1}
\begin{myremark}
The relationship between the variance and the sample size in Theorem 5.1 is well-known for an independent sample set but holds for a correlated sample set only if $nd\gg\tau$. This requirement is satisfied in the data presented in tables~\ref{tab:variance(n,d)} and \ref{tab:tau(n,d)}. If $d\gg\tau$, making the samples in $Y_b$ independent, $\sigma^2(\overline y_b)$ is always inversely proportional to $n$ even if $n$ is small, which can be understood from the definition of Eq.~\eqref{eq:variance=} where $\gamma_t\equiv0, t>0$, for independent sample sets.
\end{myremark}

\begin{theorem}
Given two sample sets with the same sample size $n$ but different sampling intervals $d_1$ and $d_2$ ($d_2>d_1$), respectively. If $nd_1\gg\tau$, $nd_2\gg\tau$ and $d_1, d_2\ll\tau$, their variances satisfy \[\dfrac{d_1}{d_2}<\dfrac{\sigma^2\left(n,d_2\right)}{\sigma^2\left(n,d_1\right)}\leq1.\]
\end{theorem}
\begin{proof}
We first discuss two sample sets:
$Y_a=\left\{y_{a,i}\left|y_{a,i}=x_i,i=1,2,\cdots,nd\right.\right\}$ containing $n\times
d$ samples and
$Y_b=\left\{y_{b,i}\left|y_{b,i}=x_{(i-1)d+1},i=1,2,\cdots,n\right.\right\}$
containing $n$ samples generated once from each $d$ samples of
$Y_a$. From Eq.~\eqref{eq:variance=}, we have:
\begin{equation}\label{eq:proof2-1}
    \begin{cases}\sigma^2(\overline y_a)=\dfrac{1}{(nd)^2}\sum_{i,j=1}^{nd}\gamma_{i,j} \\
    \sigma^2(\overline y_b)=\dfrac{1}{n^2}\sum_{i=(k_i-1)d+1 \atop {j=(k_j-1)d+1 \atop k_i,k_j=1,\cdots,n}}\gamma_{i,j}
    \end{cases}
\end{equation}
As shown in Fig.~\ref{fig:summation model}, $\sum_{i,j=1}^{nd}$ is a summation over all vertexes (without repeating) of small black quadrilaterals, but $\sum_{i=(k_i-1)d+1 \atop {j=(k_j-1)d+1 \atop k_i,k_j=1,\cdots,n}}$ is a summation over only the bottom-left vertexes of larger quadrilaterals, which are marked by red and blue colors and have indexes $k_i, k_j\in[1, n]$.

\begin{figure}
\centering
  \includegraphics[width=0.49\textwidth]{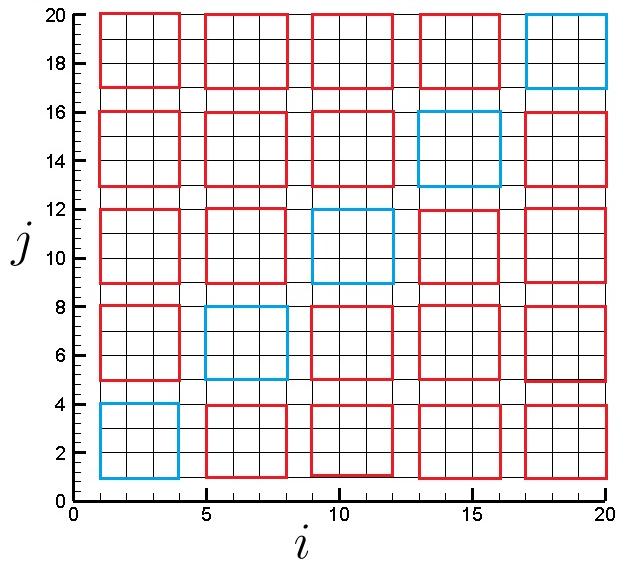}
  \includegraphics[width=0.49\textwidth]{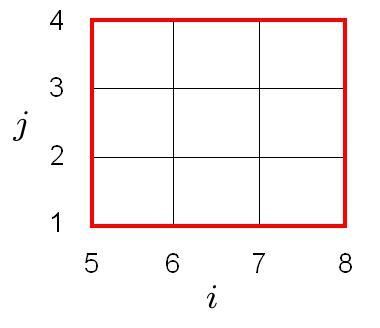}
  \caption{Schematic model for the summations with $d=4$ (left) and linear interpolation model (right).}
\label{fig:summation model}
\end{figure}

In the area of each blue quadrilateral, $\Omega_{\rm blue}$, centered at the maximal value of $\gamma_t$ (see Fig.~\ref{fig:summation model}), it can be observed that
\begin{equation}\label{eq:proof2-2}
    d\sum_{i=(k_i-1)d+1 \atop {j=(k_j-1)d+1 \atop k_i,k_j\in\Omega_{\rm blue}}}\gamma_{i,j}\leq\sum_{i,j\in\Omega_{\rm blue}}\gamma_{i,j}
    <d^2\sum_{i=(k_i-1)d+1 \atop {j=(k_j-1)d+1 \atop k_i,k_j\in\Omega_{\rm blue}}}\gamma_{i,j},
\end{equation}
where the equality holds when $\gamma_t\equiv0,\ t>0$. This can be understood by considering one of the blue quadrilaterals in Fig.~\ref{fig:summation model} (left), while realizing that the leftmost summation, $\sum_{i=(k_i-1)d+1 \atop {j=(k_j-1)d+1 \atop k_i,k_j\in\Omega_{\rm blue}}}\gamma_{i,j}$, only contains the bottom-left corner of the blue quadrilateral, namely a maximal value which lies on the diagonal. Multiplying this maximal value by $d$ will be lower or equal to $\sum_{i,j\in\Omega_{\rm blue}}\gamma_{i,j}$, having $d$ maximal terms and other terms with smaller but still positive values. The second part of the inequality stems from the fact that $d^2$ is multiplying one maximal term, and this will always be greater than $\sum_{i,j\in\Omega_{\rm blue}}\gamma_{i,j}$, having $d^2$ terms, with only $d$ terms taking maximal values.

In the area $\Omega_{\rm red}$ of those red quadrilaterals located always at the monotone interval of $\gamma_t$, we assume $d$ is much smaller than $\tau$, and thus the linear interpolation is valid in each small local red area of size $d$. For the representative red quadrilateral shown in Fig.~\ref{fig:summation model} (right) with $k_i=2$ and $k_j=1$, we have $\gamma_{5,1}=\gamma_{6,2}=\gamma_{7,3}=\gamma_{8,4}$. According to linear interpolation, we get $\gamma_{6,1}+\gamma_{5,2}=\gamma_{7,2}+\gamma_{6,3}=\gamma_{8,3}+\gamma_{7,4}\approx2\gamma_{5,1}$, $\gamma_{7,1}+\gamma_{5,3}=\gamma_{8,2}+\gamma_{6,4}\approx2\gamma_{5,1}$ and $\gamma_{8,1}+\gamma_{5,4}\approx2\gamma_{5,1}$. Thus, we have the following estimate
\begin{equation}\label{eq:proof2-3a}
    \sum_{i=5,\cdots,8\atop j=1,\cdots,4}\gamma_{i,j}\approx4^2\sum_{i=(k_i-1)4+1\atop{j=(k_j-1)4+1\atop k_i=2, k_j=1}}\gamma_{i,j}= 4^2\gamma_{5,1}.
\end{equation}
Generally, the following approximation for any arbitrary red quadrilateral is valid:
\begin{equation}\label{eq:proof2-3}
    \sum_{i,j\in\Omega_{\rm red}}\gamma_{i,j}\approx d^2\sum_{i=(k_i-1)d+1 \atop {j=(k_j-1)d+1 \atop k_i,k_j\in\Omega_{\rm red}}}\gamma_{i,j}.
\end{equation}
According to Eqs.~\eqref{eq:proof2-2}-\eqref{eq:proof2-3}, we have
\begin{align}\label{eq:proof2-4}
    d\sum_{i=(k_i-1)d+1 \atop {j=(k_j-1)d+1 \atop k_i,k_j\in\Omega_{\rm blue,all}}}\gamma_{i,j}+
    d^2\sum_{i=(k_i-1)d+1 \atop {j=(k_j-1)d+1 \atop k_i,k_j\in\Omega_{\rm red,all}}}\gamma_{i,j}\leq
    \sum_{i,j=1}^{nd}\gamma_{i,j}<d^2\sum_{i=(k_i-1)d+1 \atop {j=(k_j-1)d+1 \atop k_i,k_j=1,\cdots,n}}\gamma_{i,j}.
\end{align}
At this point, we assume Eq.~\eqref{eq:assumption} is valid and apply it to the sample set $Y_b$. Considering that $d>1$ and $d^2>d$, we get
\begin{equation}\label{eq:proof2-5}
    d\sum_{i=(k_i-1)d+1 \atop {j=(k_j-1)d+1 \atop k_i,k_j=1,\cdots,n}}\gamma_{i,j}\leq
    \sum_{i,j=1}^{nd}\gamma_{i,j}<d^2\sum_{i=(k_i-1)d+1 \atop {j=(k_j-1)d+1 \atop k_i,k_j=1,\cdots,n}}\gamma_{i,j},
\end{equation}
where the equality holds when $\gamma_t\equiv0,\ t>0$. Substituting Eq.~\eqref{eq:proof2-5} into Eq.~\eqref{eq:proof2-1}, we have
\begin{equation}\label{eq:proof2-6}
    \dfrac{\sigma^2\left(\overline y_b\right)}{d}\leq\sigma^2(\overline y_a)<\sigma^2(\overline y_b).
\end{equation}
We introduce $Y_c=\{y_{c,i}|y_{c,i}=x_i,i=1,2,\cdots,n\}$ which contains $n$ samples as $Y_b$ and has the same correlation degree as $Y_a$. Since $n\gg\tau$ and $nd\gg\tau$ according to the assumption, the conclusion of Eq. \eqref{eq:proof1-2} implies that
\begin{equation}\label{eq:proof2-add}
    \dfrac{\sigma^2\left(\overline y_c\right)}{\sigma^2\left(\overline y_a\right)}=\dfrac{nd}{n}=d.
\end{equation}
Substituting Eq. \eqref{eq:proof2-add} in Eq. \eqref{eq:proof2-6}, we get
\begin{equation}\label{eq:proof2-7}
    \dfrac{1}{d}<\dfrac{\sigma^2\left(\overline
      y_b\right)}{\sigma^2\left(\overline y_c\right)}\leq1,
\end{equation}
with which the proof is complete.
\end{proof}
\newtheorem*{myremarkb}{Remark 5.2}
\begin{myremarkb}
Taking the sample set with $d_1$ in Eq.~\eqref{eq:obser2} as $Y_c$ and
the other as $Y_b$, we observe that Eq.~\eqref{eq:obser2} is
equivalent to Eq.~\eqref{eq:proof2-7} proved here. If the
sample set $Y_a$ (namely $Y_c$) has a high correlation degree, the summation over area $\Omega_{\rm red}$ is dominant (see
Fig.~\ref{fig:correlation degree}~(left)) and \(\dfrac{\sigma^2\left(\overline y_b\right)}{\sigma^2\left(\overline y_c\right)}\) converges to $\dfrac{1}{d}$ according to Eq.~\eqref{eq:proof2-3} which implies $\sigma^2\left(\overline y_a\right)=\sigma^2\left(\overline y_b\right)$. In contrast, $\dfrac{\sigma^2(\overline y_b)}{\sigma^2(\overline y_c)}=1$ if the samples in $Y_a$ are independent.

The assumptions of theorem 5.2 are that $nd_1\gg\tau$, $nd_2\gg\tau$ and $d_1, d_2\ll\tau$, which are satisfied in the data shown in tables~\ref{tab:variance(n,d)} and \ref{tab:tau(n,d)}. In real applications, $nd$ should be much larger than $\tau$ since otherwise the variance of the average value is very high, which makes the average value not trustworthy. For the selection of $d$, we suggest to let $d$ be much larger than 1 to reduce memory usage. In addition, we also suggest to let $d$ be much smaller than $\tau$ as otherwise this leads to loss of too much correlated information that can still reduce the variance effectively. The two necessary assumptions can thus be easily satisfied in real applications.
\end{myremarkb}

\section{Conclusions}\label{s:conc}
The influence of the sample size $n$ and sampling interval $d$ used in MCMC simulations on the variance of the average quantities is analyzed using numerical results and proved theoretically. If $n\times d$ is much larger than the maximal correlation interval $\tau$ of the full sample set, the variance of the estimation using a coarse sample set with fixed $d$ is inversely proportional to $n$ and the CPU time. For a given CPU time, the memory or disk usage (namely the sample size) can be reduced greatly by increasing $d$, while getting a negligible increase in variance if the original $d$ is very small.

In the implementation of the blocking method, the blocking process is subject to increased fluctuations when the sample size $n'$ is reduced; in particular, the fluctuation gets worse when $n'$ approaches two. The current results show that the fluctuation starts near the first maximal point obtained during the blocking process. Additionally, the corresponding maximal value can be used as an estimate of the variance if the blocking process converges, or as a lower bound estimate of the variance if the blocking process does not converge.

\section*{Acknowledgments}
This work was supported in part by the King Abdullah University of Science and Technology (KAUST) Center for Numerical Porous Media.  In addition, S. Sun would also like to acknowledge the support of this study by a research award from King Abdulaziz City for Science and Technology (KACST) through a project entitled "Study of Sulfur Solubility using Thermodynamics Model and Quantum ChemistryÓ. 


\end{document}